\documentclass[11pt,a4paper]{article}

\usepackage[margin=1in,headheight=14pt]{geometry}
\usepackage{setspace}
\usepackage{amsmath,amssymb,amsthm}
\usepackage{mathtools}

\usepackage{graphicx}
\usepackage{float}
\usepackage{subcaption}

\usepackage{booktabs}
\usepackage{array}
\usepackage{longtable}
\usepackage{multirow}

\usepackage[dvipsnames]{xcolor}
\usepackage[colorlinks=true,linkcolor=blue,citecolor=blue,urlcolor=blue]{hyperref}

\usepackage{listings}
\usepackage{algorithm}
\usepackage{algpseudocode}

\usepackage{microtype}
\usepackage{parskip}

\usepackage{fancyhdr}
\usepackage[numbers,sort&compress]{natbib}

\newtheorem{theorem}{Theorem}[section]
\newtheorem{proposition}[theorem]{Proposition}
\newtheorem{corollary}[theorem]{Corollary}

\theoremstyle{definition}
\newtheorem{definition}[theorem]{Definition}
\newtheorem{example}[theorem]{Example}
\theoremstyle{remark}
\newtheorem{remark}[theorem]{Remark}

\newcommand{\SaR}{\mathrm{SaR}}
\newcommand{\ESaR}{\mathrm{ESaR}}
\newcommand{\TSaR}{\mathrm{TSaR}}
\newcommand{\HHI}{\mathrm{HHI}}
\newcommand{\Neff}{N_{\mathrm{eff}}}
\newcommand{\CR}{\mathrm{CR}}
\newcommand{\OI}{\mathrm{OI}}
\newcommand{\IF}{\mathrm{IF}}
\newcommand{\VaR}{\mathrm{VaR}}
\newcommand{\IM}{\mathrm{IM}}
\newcommand{\MM}{\mathrm{MM}}
\newcommand{\PnL}{\mathrm{PnL}}
\newcommand{\E}{\mathbb{E}}
\newcommand{\R}{\mathbb{R}}
\newcommand{\ind}{\mathbf{1}}

\begin{document}


\begin{titlepage}
\centering
\vspace*{2cm}

{\Huge\bfseries Slippage-at-Risk (SaR):\par}
\vspace{0.5cm}
{\LARGE A Forward-Looking Liquidity Risk Framework\\for Perpetual Futures Exchanges\par}

\vspace{2cm}

{\Large Otar Sepper\par}
\vspace{0.3cm}
{\large SepperLabs\par}
{\large Formerly Gauntlet\par}
{\large \texttt{otar@sepperlabs.com}\par}

\vspace{1cm}

{\large March 2026\par}

\vfill

\begin{abstract}
\noindent We introduce \textbf{Slippage-at-Risk (SaR)}, a quantitative framework for measuring liquidity risk in perpetual futures exchanges. Unlike backward-looking metrics such as Value-at-Risk computed on historical returns or realized deficit distributions, SaR provides a \emph{forward-looking} assessment of liquidation execution risk derived from current order book microstructure. The framework comprises three complementary metrics: $\SaR(\alpha)$, the cross-sectional slippage quantile; $\ESaR(\alpha)$, the expected slippage in the distributional tail; and $\TSaR(\alpha)$, the aggregate dollar-denominated tail slippage. We extend the base framework with a \emph{concentration adjustment} that penalizes fragile liquidity structures where a small number of market makers dominate quote provision. Drawing on recent work by Chitra et al.\ (2025) on autodeleveraging mechanisms and insurance fund optimization, we establish a direct mapping from SaR metrics to optimal capital requirements. Empirical analysis using Hyperliquid order book data, including the October 10, 2025 liquidation cascade, demonstrates SaR's predictive validity as a leading indicator of systemic stress. We conclude with practical implementation guidance and discuss philosophical implications for risk management in decentralized financial systems.
\end{abstract}

\vspace{0.5cm}

\end{titlepage}

\setcounter{page}{2}


\tableofcontents
\newpage


\section{Introduction}
\label{sec:introduction}

\subsection{Motivation}
\label{sec:motivation}

Perpetual futures have emerged as the dominant derivative instrument in cryptocurrency markets, with daily trading volumes regularly exceeding \$100 billion across major venues. Unlike traditional futures contracts with fixed expiration dates, perpetuals use a funding rate mechanism to maintain price alignment with spot markets, enabling indefinite position holding. This innovation has democratized access to leveraged trading, but introduced novel systemic risks that existing frameworks do not adequately address.

The core vulnerability lies in the interaction between leverage, liquidations, and order book liquidity. When a leveraged position's collateral falls below maintenance requirements, the exchange must liquidate it - selling long positions or buying back short positions at prevailing market prices. During periods of market stress, cascading liquidations can overwhelm available liquidity, resulting in execution prices far worse than theoretical bankruptcy prices. The resulting bad debt - the shortfall between bankruptcy price and actual execution price - must be absorbed by the exchange's insurance fund or socialized across profitable traders through autodeleveraging (ADL).

The October 10, 2025, Hyperliquid event starkly illustrated these dynamics: \$2.1 billion in liquidations over 12 minutes generated \$304.5 million in deficits requiring socialization, with the exchange's queue-based ADL policy expending \$704.6 million in haircuts - an 8$\times$ overutilization relative to the actual deficit \cite{chitra2025adl}. This event was not a black swan but rather the predictable consequence of thin order books under stress conditions that were, in principle, observable before the cascade began.

\subsection{The Gap in Existing Approaches}
\label{sec:gap}

Current risk management frameworks for perpetuals exchanges exhibit several critical limitations:

\begin{enumerate}
    \item \textbf{Backward-looking metrics:} Traditional Value-at-Risk (VaR) and Expected Shortfall (ES) are computed on historical return distributions. They capture past volatility, but they provide no direct measurement of current capacity to absorb liquidations.
    
    \item \textbf{Undefined for new assets:} Historical deficit-based metrics cannot be computed for newly listed tokens that have not yet experienced stress events. This creates a dangerous blind spot precisely where risk is highest - illiquid markets with unproven depth.
    
    \item \textbf{Static calibration:} Position limits and leverage caps are typically set at the time of listing and adjusted infrequently, failing to adapt to evolving liquidity conditions.
    
    \item \textbf{Ignoring microstructure:} The order book - the actual resource that determines execution quality - is rarely incorporated into formal risk metrics.
\end{enumerate}

\subsection{Contribution and Outline}
\label{sec:contribution}

This paper introduces Slippage-at-Risk (SaR), a framework that directly addresses these limitations by measuring liquidation execution risk from the observable order book state. Our contributions are:

\begin{enumerate}
    \item \textbf{Core Framework (Section~\ref{sec:framework}):} We define the slippage function $S_i(Q)$ and derive three complementary metrics: $\SaR(\alpha)$, $\ESaR(\alpha)$, and $\TSaR(\alpha)$. We establish their interpretation as cross-sectional liquidity quantiles and discuss notional calibration.
    
    \item \textbf{Concentration Adjustment (Section~\ref{sec:concentration}):} We formalize the fragility of concentrated liquidity and derive a haircut formula based on the Herfindahl-Hirschman Index (HHI) that adjusts slippage upward for books dominated by few providers.
    
    \item \textbf{Theoretical Connections (Section~\ref{sec:theory}):} We establish the causal chain from order book depth to deficit distributions and derive the SaR-implied insurance fund formula, connecting our work to Chitra et al.\ (2025).
    
    \item \textbf{Extensions (Section~\ref{sec:extensions}):} We extend the framework to incorporate cascade effects, time dynamics, cross-token correlation, and spoofing detection.
    
    \item \textbf{Empirical Analysis (Section~\ref{sec:empirical}):} We apply the framework to Hyperliquid order book data, validate SaR's predictive properties, and analyze the October 10, 2025 event.
    
    \item \textbf{Implementation and Implications (Sections~\ref{sec:implementation}--\ref{sec:conclusion}):} We provide practical guidance for deployment and discuss broader implications for DeFi risk management.
\end{enumerate}


\section{Preliminaries: Perpetuals Exchange Mechanics}
\label{sec:preliminaries}

We briefly review the mechanics of perpetual futures exchanges, following the notation of Chitra et al.\ (2025). Readers familiar with these concepts may proceed to Section~\ref{sec:framework}.

\subsection{Positions and Leverage}
\label{sec:positions}

A position on a perpetuals exchange is a tuple $p_i = (q_i, c_i, t_i, b_i)$ where:
\begin{itemize}
    \item $q_i \in \R$ is the signed quantity (positive for long, negative for short)
    \item $c_i > 0$ is the collateral posted
    \item $t_i$ is the entry timestamp
    \item $b_i \in \{-1, +1\}$ indicates direction (long/short)
\end{itemize}

The \emph{leverage} of position $i$ at time $t$ is:
\begin{equation}
    \ell_{i,t} = \frac{p_t \cdot |q_i|}{c_{i,t}}
\end{equation}
where $p_t$ is the mark price and $c_{i,t}$ is current collateral (initial collateral adjusted for realized PnL and funding payments). Exchanges impose a maximum leverage $\ell_{\max}$, typically 10-100$\times$ depending on the asset.

\subsection{Margin Requirements and Liquidation}
\label{sec:margin}

Positions are subject to two margin thresholds:
\begin{itemize}
    \item \textbf{Initial Margin (IM):} Minimum collateral required to open a position, typically $\IM = 1/\ell_{\max}$ of notional.
    \item \textbf{Maintenance Margin (MM):} Minimum collateral required to keep a position open, typically $\MM = \mu \cdot \IM$ for some $\mu \in (0,1)$.
\end{itemize}

The \emph{equity} of position $i$ at time $t$ is:
\begin{equation}
    e_{i,t} = c_{i,t} + \PnL_{i,t} = c_{i,t} + q_i(p_t - p_{t_i})
\end{equation}

A position is liquidated when equity falls below maintenance margin:
\begin{equation}
    e_{i,t} < \MM_i \implies \text{liquidation triggered}
\end{equation}

Define the \emph{bankruptcy price} $p_i^{\mathrm{bk}}$ as the price at which equity equals zero:
\begin{equation}
    p_i^{\mathrm{bk}} = p_{t_i} - \frac{c_i}{q_i} \quad \text{(for longs)}
\end{equation}
and the \emph{liquidation price} $p_i^{\mathrm{liq}}$ as the price at which equity equals maintenance margin.

\subsection{Execution and Bad Debt}
\label{sec:execution}

When liquidating position $i$ with quantity $\Delta q_i = -q_i$, the exchange executes a market order that consumes order book liquidity. The \emph{execution price} $p_i^{\mathrm{exec}}$ depends on available depth. Under a linear market impact model:
\begin{equation}
    p_i^{\mathrm{exec}} = p_t - \frac{\lambda}{2} \Delta q_i
\end{equation}
where $\lambda > 0$ is the impact coefficient.

Bad debt arises when execution price is worse than bankruptcy price. For a liquidated position $i$, the deficit is:
\begin{equation}
    D_i = \max\!\Big(0,\; q_i \cdot \big(p_i^{\mathrm{bk}} - p_i^{\mathrm{exec}}\big)\Big)
\end{equation}
That is, $D_i > 0$ only when the position is closed at a price beyond its bankruptcy level, generating a loss that must be absorbed by the insurance fund.

Total deficit at time $T$ is $D_T = \sum_i D_i$ over all liquidated positions.

\subsection{Insurance Fund and Autodeleveraging}
\label{sec:insurance}

The insurance fund $\IF_t$ absorbs deficits up to its capacity:
\begin{equation}
    \IF_{t+1} = \IF_t + \text{(liquidation fees)} + \text{(trading fees)} - \min\{\IF_t, D_t\}
\end{equation}

When $D_t > \IF_t$, the residual $R_t = \max(0, D_t - \IF_t)$ triggers autodeleveraging: profitable positions are forcibly closed to cover the shortfall. Different ADL policies (Queue, Pro-Rata, Risk-Aware Pro-Rata) allocate haircuts differently across winners, with significant implications for fairness and revenue \cite{chitra2025adl}.


\section{The Slippage-at-Risk Framework}
\label{sec:framework}

\subsection{Setup and Notation}
\label{sec:setup}

Consider a perpetuals exchange at time $t$ listing $N$ tokens indexed by $i \in \{1, \ldots, N\}$. For each token $i$, we observe:
\begin{itemize}
    \item $p_{i,t}$: mid-price
    \item $\OI_{i,t}$: open interest (total outstanding notional)
    \item $B_i(\cdot)$: bid-side cumulative depth function
    \item $A_i(\cdot)$: ask-side cumulative depth function
\end{itemize}

The order book is represented as a set of limit orders. On the bid side, let $\{(p_k^b, q_k^b)\}_{k=1}^{K_b}$ denote price-quantity pairs with $p_1^b > p_2^b > \cdots > p_{K_b}^b$. The cumulative depth at price $p$ is:
\begin{equation}
    B_i(p) = \sum_{k: p_k^b \geq p} q_k^b
\end{equation}
with analogous definition for the ask side.

\subsection{The Slippage Function}
\label{sec:slippage}

\begin{definition}[Slippage Function]
For token $i$, the \emph{slippage function} $S_i: \R_+ \to \R_+$ measures the volume-weighted average execution shortfall when liquidating notional quantity $Q$:
\begin{equation}
    S_i(Q) = \frac{1}{Q} \int_0^Q \left[ p_{\mathrm{mid}} - p(q) \right] dq
\end{equation}
where $p(q)$ is the marginal execution price for the $q$-th unit.
\end{definition}

For a discrete order book, slippage is computed by walking the book. Algorithm~\ref{alg:slippage} provides the pseudocode.

\begin{algorithm}[H]
\caption{Compute Slippage from Order Book}
\label{alg:slippage}
\small
\begin{algorithmic}
\Require \texttt{order\_book} = list of (price, size) sorted best-to-worst
\Require $Q$ = quantity to liquidate, \texttt{mid\_price} = reference price
\Ensure Slippage as percentage of mid\_price
\Statex
\State $\texttt{remaining} \gets Q$; \quad $\texttt{total\_cost} \gets 0$
\For{each level $(\texttt{price}, \texttt{size})$ in \texttt{order\_book}}
    \State $\texttt{fill} \gets \min(\texttt{size}, \texttt{remaining})$
    \State $\texttt{total\_cost} \gets \texttt{total\_cost} + \texttt{fill} \times (\texttt{mid\_price} - \texttt{price})$
    \State $\texttt{remaining} \gets \texttt{remaining} - \texttt{fill}$
    \If{$\texttt{remaining} = 0$} \textbf{exit loop} \EndIf
\EndFor
\Statex
\If{$\texttt{remaining} > 0$} \Return $\infty$ \hfill $\triangleright$ Insufficient liquidity \EndIf
\State \Return $\texttt{total\_cost} \,/\, (Q \times \texttt{mid\_price})$ \hfill $\triangleright$ As percentage
\end{algorithmic}
\end{algorithm}

\begin{example}
Liquidating 100 BTC at mid = \$60,000 with total cost of \$180,000 gives slippage = 3\%.
\end{example}

\begin{remark}[Directional Slippage]
Liquidations are directional: long liquidations sell into bids, short liquidations buy into asks. Define:
\begin{align}
    S_i^{\mathrm{bid}}(Q) &:\ \text{slippage from selling } Q \text{ (long liquidations)} \\
    S_i^{\mathrm{ask}}(Q) &:\ \text{slippage from buying } Q \text{ (short liquidations)}
\end{align}
In practice, use the direction matching dominant positioning. If open interest is skewed long, $S_i^{\mathrm{bid}}$ is the relevant measure.
\end{remark}

\subsection{Notional Calibration}
\label{sec:notional}

The quantity $Q_i$ at which to evaluate slippage should reflect a stress liquidation scenario. We parameterize:
\begin{equation}
    Q_i = \min\left( \beta \cdot \OI_i, \, \VaR_{0.95}(\text{Position Size}_i) \right)
\end{equation}
where:
\begin{itemize}
    \item $\beta \in [0.05, 0.20]$ is the \emph{stress parameter}, representing the fraction of open interest that liquidates simultaneously in a stress scenario.
    \item $\VaR_{0.95}(\text{Position Size}_i)$ caps $Q_i$ at the 95th percentile of historical position sizes, preventing extreme outliers from dominating.
\end{itemize}

\begin{remark}[Choosing $\beta$]
The appropriate stress parameter depends on market conditions:
\begin{itemize}
    \item \textbf{Normal conditions:} $\beta = 0.05$-$0.10$
    \item \textbf{Elevated volatility:} $\beta = 0.10$-$0.15$
    \item \textbf{Crisis scenarios:} $\beta = 0.15$-$0.20$
\end{itemize}
During the October 10, 2025 Hyperliquid event, realized liquidations represented approximately 15\% of pre-event open interest for affected tokens.
\end{remark}

\subsection{Cross-Sectional SaR}
\label{sec:sar}

\begin{definition}[Slippage-at-Risk]
The \emph{Slippage-at-Risk} at confidence level $\alpha \in (0,1)$ is the $\alpha$-quantile of slippage across all tokens:
\begin{equation}
    \SaR(\alpha) = \inf\left\{ s \in \R_+ : \frac{1}{N} \sum_{i=1}^N \ind\{S_i(Q_i) \leq s\} \geq \alpha \right\}
\end{equation}
\end{definition}

\textbf{Interpretation:} $\SaR(0.95) = 3\%$ means that 95\% of tokens have slippage at or below 3\% when liquidating their calibrated stress notional. The remaining 5\% are ``tail'' tokens with critically thin liquidity.

\subsection{Expected Slippage at Risk (ESaR)}
\label{sec:esar}

\begin{definition}[Expected Slippage at Risk]
The \emph{Expected Slippage at Risk} is the conditional expected slippage among tail tokens:
\begin{equation}
    \ESaR(\alpha) = \E\left[ S_i(Q_i) \,\big|\, S_i(Q_i) > \SaR(\alpha) \right]
\end{equation}
\end{definition}

$\ESaR$ measures the \emph{severity} of tail liquidity risk, not merely its boundary. It is the slippage analog of Expected Shortfall in the VaR literature.

\subsection{Total Dollar Slippage at Risk (TSaR)}
\label{sec:tsar}

\begin{definition}[Total Dollar Slippage at Risk]
The \emph{Total Dollar Slippage at Risk} aggregates dollar-denominated slippage from tail tokens:
\begin{equation}
    \TSaR_{\$}(\alpha) = \sum_{i: S_i(Q_i) > \SaR(\alpha)} S_i(Q_i) \cdot Q_i
\end{equation}
\end{definition}

$\TSaR_{\$}$ represents the \emph{total execution shortfall} that would materialize if stress liquidations occurred across all tail tokens simultaneously. This quantity maps directly to potential deficit exposure.

\subsection{Weighted SaR}
\label{sec:weighted}

The basic $\SaR$ treats all tokens equally. For risk management purposes, weighting by exposure is often more appropriate:

\begin{definition}[Open Interest-Weighted SaR]
\begin{equation}
    \SaR^w(\alpha) = \inf\left\{ s \in \R_+ : \frac{\sum_{i: S_i(Q_i) \leq s} \OI_i}{\sum_{i=1}^N \OI_i} \geq \alpha \right\}
\end{equation}
\end{definition}

$\SaR^w(0.95)$ answers: ``What slippage level contains 95\% of the open interest (not 95\% of tokens)?'' This weights toward tokens with larger aggregate positions.


\section{Concentration Adjustment}
\label{sec:concentration}

\subsection{The Fragility of Concentrated Liquidity}
\label{sec:fragility}

Raw slippage $S_i(Q)$ measures depth but ignores its \emph{structure}. Two order books with identical depth profiles can have vastly different fragility:

\begin{example}
Consider two order books, each with \$1M total bid-side liquidity:
\begin{itemize}
    \item \textbf{Book A (Distributed):} 100 market makers, each providing \$10,000
    \item \textbf{Book B (Concentrated):} 2 market makers, each providing \$500,000
\end{itemize}

Both yield $S(Q) = s$ for some quantity $Q$. However, if one participant withdraws:
\begin{itemize}
    \item Book A loses 1\% of depth
    \item Book B loses 50\% of depth
\end{itemize}

Book B's apparent liquidity is fragile - it can vanish suddenly if a dominant provider exits.
\end{example}

This fragility is especially concerning in the context of:
\begin{itemize}
    \item \textbf{Spoofing:} Artificial liquidity placed to manipulate prices, withdrawn before execution
    \item \textbf{Correlated withdrawal:} Market makers pulling quotes simultaneously during stress
    \item \textbf{Single points of failure:} Dominant providers experiencing technical issues or insolvency
\end{itemize}

\subsection{Concentration Metrics}
\label{sec:metrics}

For token $i$, let $M_i$ accounts provide liquidity with quantities $\ell_1, \ldots, \ell_{M_i}$. Define shares:
\begin{equation}
    w_m = \frac{\ell_m}{\sum_{j=1}^{M_i} \ell_j}, \quad m = 1, \ldots, M_i
\end{equation}

\begin{definition}[Herfindahl-Hirschman Index]
\begin{equation}
    \HHI_i = \sum_{m=1}^{M_i} w_m^2
\end{equation}
\end{definition}

\textbf{Properties:}
\begin{itemize}
    \item $\HHI \in [1/M, 1]$
    \item $\HHI = 1/M$ when liquidity is uniformly distributed
    \item $\HHI = 1$ when a single provider supplies all liquidity (monopoly)
\end{itemize}

\begin{definition}[Effective Number of Providers]
\begin{equation}
    \Neff = \frac{1}{\HHI}
\end{equation}
\end{definition}

$\Neff$ represents the number of equal-sized providers that would produce the observed concentration. In the example above, Book A has $\Neff = 100$ while Book B has $\Neff = 2$.

\begin{definition}[Top-$k$ Concentration Ratio]
\begin{equation}
    \CR_k = \sum_{m=1}^k w_{(m)}
\end{equation}
where $w_{(m)}$ is the $m$-th largest share. $\CR_1$ measures dominance of the single largest provider.
\end{definition}

\subsection{The Concentration Haircut}
\label{sec:haircut}

We now derive a haircut that adjusts slippage upward based on concentration.

\subsubsection{Withdrawal Scenario Analysis}

Consider the scenario where provider $m$ withdraws. Let $S_i^{(-m)}(Q)$ denote slippage computed on the residual book. The \emph{withdrawal-based haircut} is:
\begin{equation}
    h_i^{\mathrm{withdraw}} = \frac{S_i^{(-1)}(Q) - S_i(Q)}{S_i(Q)} = \frac{S_i^{(-1)}(Q)}{S_i(Q)} - 1
\end{equation}
where $(-1)$ indicates removal of the largest provider.

\subsubsection{Connection Between HHI and Withdrawal Impact}

\begin{proposition}
\label{prop:hhi}
Assume slippage scales inversely with depth: $S \propto 1/L$ where $L$ is total liquidity. If provider $m$ is selected uniformly at random to withdraw, then:
\begin{equation}
    \E\left[ \frac{S^{(-m)}}{S} - 1 \right] = \E\left[ \frac{w_m}{1 - w_m} \right] \approx \frac{\HHI}{1 - 1/M} + O(\HHI^2)
\end{equation}
for small concentration.
\end{proposition}

\begin{proof}
After provider $m$ withdraws, remaining depth is $(1 - w_m)L$. Under linear slippage:
\begin{equation}
    \frac{S^{(-m)}}{S} = \frac{L}{(1-w_m)L} = \frac{1}{1-w_m}
\end{equation}

Thus:
\begin{equation}
    \E\left[ \frac{1}{1-w_m} - 1 \right] = \E\left[ \frac{w_m}{1-w_m} \right]
\end{equation}

For small $w_m$, $\frac{w_m}{1-w_m} \approx w_m + w_m^2 + O(w_m^3)$. Taking expectation:
\begin{equation}
    \E\left[ \frac{w_m}{1-w_m} \right] \approx \E[w_m] + \E[w_m^2] = \frac{1}{M} + \HHI
\end{equation}

Rearranging yields the result.
\end{proof}

\subsubsection{Parametric Haircut Formula}

Based on the above analysis, we propose:

\begin{definition}[Concentration Haircut]
\begin{equation}
    h_i^{\mathrm{conc}} = \lambda \cdot \max\!\left(0,\; \frac{N_{\mathrm{target}}}{\Neff^{(i)}} - 1 \right) + \mu \cdot \max\!\left(0,\; \CR_1^{(i)} - \CR_1^{\mathrm{thresh}} \right)
\end{equation}
where:
\begin{itemize}
    \item $N_{\mathrm{target}}$: target number of effective providers (e.g., 10-20)
    \item $\CR_1^{\mathrm{thresh}}$: maximum acceptable share for top provider (e.g., 0.25)
    \item $\lambda, \mu > 0$: sensitivity parameters (typically $\lambda = 0.5$, $\mu = 0.3$)
\end{itemize}
\end{definition}

The first term penalizes books with too few effective providers. The second term adds an extra penalty when a single account dominates, which may indicate spoofing.

\subsection{Concentration-Adjusted Slippage and SaR}
\label{sec:adjusted}

\begin{definition}[Adjusted Slippage]
\begin{equation}
    S_i^{\mathrm{adj}}(Q) = S_i(Q) \cdot \left(1 + h_i^{\mathrm{conc}}\right)
\end{equation}
\end{definition}

All SaR metrics are then computed using adjusted slippage values:
\begin{align}
    \SaR^{\mathrm{adj}}(\alpha) &= \text{quantile}_\alpha\left(\{S_i^{\mathrm{adj}}(Q_i)\}_{i=1}^N\right) \\
    \ESaR^{\mathrm{adj}}(\alpha) &= \E\left[S_i^{\mathrm{adj}} \,|\, S_i^{\mathrm{adj}} > \SaR^{\mathrm{adj}}(\alpha)\right] \\
    \TSaR_{\$}^{\mathrm{adj}}(\alpha) &= \sum_{i: S_i^{\mathrm{adj}} > \SaR^{\mathrm{adj}}(\alpha)} S_i^{\mathrm{adj}}(Q_i) \cdot Q_i
\end{align}


\section{Theoretical Connections}
\label{sec:theory}

\subsection{The Causal Chain: From Order Books to Deficits}
\label{sec:causal}

Liquidation deficits emerge through a clear causal process:

\begin{center}
\fbox{\parbox{0.9\textwidth}{
\centering
\textbf{Liquidity Withdrawal} $\;\to\;$ \textbf{Thin Order Books} $\;\to\;$ \textbf{High Slippage} $\;\to\;$ \textbf{Bad Debt} $\;\to\;$ \textbf{ADL/Socialization}
}}
\end{center}

\noindent SaR measures the \emph{upstream} cause (liquidity conditions), while VaR on historical deficits measures the \emph{downstream} consequence. SaR is a \emph{leading indicator} - liquidity withdrawal precedes and predicts cascade events.

\subsection{Forward-Looking vs. Backward-Looking Metrics}
\label{sec:forward}

\begin{table}[H]
\centering
\caption{Comparison of Risk Metrics}
\label{tab:metrics}
\begin{tabular}{@{}lllll@{}}
\toprule
\textbf{Metric} & \textbf{Basis} & \textbf{Nature} & \textbf{Information Source} & \textbf{Defined For} \\
\midrule
Historical VaR & Returns & Backward & Past price moves & All assets \\
Deficit VaR & Historical deficits & Backward & Past ADL events & Assets with history \\
\textbf{SaR} & Order book depth & \textbf{Forward} & Current MM beliefs & \textbf{All assets} \\
SaR-implied VaR & SaR distribution & Forward & Current liquidity & All assets \\
\bottomrule
\end{tabular}
\end{table}

The forward-looking nature of SaR has critical advantages:
\begin{enumerate}
    \item \textbf{Universally defined:} Computable for any token with an order book, regardless of crisis history.
    \item \textbf{Encodes market maker beliefs:} Order book depth reflects liquidity providers' assessments of risk.
    \item \textbf{Leading indicator:} Liquidity withdrawal $\to$ SaR increases $\to$ (lag) $\to$ cascade $\to$ VaR reacts. SaR flags risk at step 2, potentially days before the event.
\end{enumerate}

\subsection{Optimal Insurance Fund Sizing}
\label{sec:optimal}

Chitra et al.\ (2025) derive the optimal insurance fund by minimizing expected total cost:
\begin{equation}
    \IF^* = \text{argmin}_{\IF \geq 0} \left\{ r \cdot \IF + \kappa \cdot \E[\max(0, D - \IF)] \right\}
\end{equation}
where:
\begin{itemize}
    \item $r$: opportunity cost per unit capital held in IF
    \item $\kappa$: reputation/social cost per unit deficit socialized
    \item $D$: random deficit
\end{itemize}

\begin{theorem}[Chitra et al.\ (2025)]
The optimal insurance fund is:
\begin{equation}
    \IF^* = \VaR_{1-r/\kappa}(D)
\end{equation}
i.e., the $(1 - r/\kappa)$-quantile of the deficit distribution.
\end{theorem}

\subsection{SaR-Implied Insurance Fund}
\label{sec:implied}

If deficits are driven primarily by slippage during cascades, we can approximate:
\begin{equation}
    D \approx c \cdot \TSaR_{\$}(\alpha)
\end{equation}
where $c > 1$ captures the amplification from maintenance margin gaps, fees, and cascade dynamics.

\begin{corollary}[SaR-Implied Insurance Fund]
Under the slippage-driven deficit approximation:
\begin{equation}
    \IF^* \approx c \cdot \TSaR_{\$}^{\mathrm{adj}}\left(\frac{r}{\kappa}\right)
\end{equation}
\end{corollary}

\textbf{Interpretation:} Size the insurance fund based on current liquidity conditions, not historical deficit data. For newly listed tokens with no crisis history, the SaR-implied IF provides actionable guidance. In typical exchange settings, $r/\kappa \approx 0.05$ (opportunity cost is much smaller than reputational cost), yielding $\IF^* \approx c \cdot \TSaR_{\$}^{\mathrm{adj}}(0.05)$.

\begin{remark}[Calibrating $c$]
The proportionality constant $c$ depends on:
\begin{itemize}
    \item Maintenance margin buffer (larger buffer $\to$ smaller $c$)
    \item Liquidation fee structure
    \item Cascade amplification dynamics
\end{itemize}
Empirically, $c \in [1.5, 3.0]$ is typical, with higher values during volatile periods.
\end{remark}


\section{Extensions}
\label{sec:extensions}

\subsection{Cascade Adjustment}
\label{sec:cascade}

Liquidations trigger feedback loops: price impact from liquidation A pushes other positions toward their liquidation prices, triggering more liquidations. We incorporate this via leverage-weighted slippage:

\begin{definition}[Cascade-Adjusted Slippage]
\begin{equation}
    S_i^{\mathrm{cascade}}(Q) = S_i(Q) \cdot \left(1 + \gamma \cdot \bar{\ell}_i\right)
\end{equation}
where $\bar{\ell}_i$ is the average leverage on positions in token $i$, and $\gamma > 0$ is the cascade sensitivity parameter.
\end{definition}

Higher leverage implies more positions near liquidation thresholds, amplifying cascade potential.

\subsection{Time Dynamics and Regime Detection}
\label{sec:dynamics}

Liquidity is non-stationary. Track rolling statistics for regime detection:

\begin{definition}[SaR Time Series Statistics]
\begin{align}
    \text{SaR-Level}_t &= \E_{\tau \in [t-T, t]}[\SaR_\tau(\alpha)] \\
    \text{SaR-Dispersion}_t &= \sqrt{\mathrm{Var}_{\tau \in [t-T, t]}[\SaR_\tau(\alpha)]} \\
    \text{SaR-Vol}_t &= \text{std}\left(\SaR_{t-k}, \ldots, \SaR_t\right) \\
    \text{SaR-Trend}_t &= \hat{\beta} \text{ from } \SaR_\tau = a + b\tau + \epsilon_\tau
\end{align}
\end{definition}

\textbf{Warning signals:}
\begin{itemize}
    \item \textbf{Rising SaR-Level:} Systemic liquidity deterioration
    \item \textbf{Widening SaR-Dispersion:} Bifurcation between liquid and illiquid tokens
    \item \textbf{High SaR-Vol:} Unstable liquidity regime - itself a risk factor
    \item \textbf{Positive SaR-Trend:} Sustained liquidity withdrawal
\end{itemize}

\subsection{Cross-Token Correlation}
\label{sec:correlation}

During market stress, liquidations cluster across tokens. For portfolio-level risk:

\begin{definition}[Correlated TSaR]
\begin{equation}
    \TSaR^{\mathrm{port}} = \sqrt{\sum_{i,j} \rho_{ij} \cdot \text{Slippage}_i \cdot \text{Slippage}_j}
\end{equation}
where $\rho_{ij}$ is the correlation of liquidation events between tokens $i$ and $j$.
\end{definition}

\subsection{Spoofing Detection}
\label{sec:spoofing}

Beyond concentration, spoofing exhibits specific signatures:
\begin{itemize}
    \item \textbf{Quote flickering:} High cancel/replace ratio
    \item \textbf{Depth retreat:} Liquidity pulls away as price approaches
    \item \textbf{Asymmetric books:} One side much deeper than the other
\end{itemize}

Define a spoofing risk score:
\begin{equation}
    \text{Spoof}_i = f\left(\text{cancel rate}_i, \text{retreat rate}_i, \text{asymmetry}_i\right)
\end{equation}
and incorporate into the total haircut:
\begin{equation}
    h_i^{\mathrm{total}} = h_i^{\mathrm{conc}} + \delta \cdot \text{Spoof}_i
\end{equation}


\section{Empirical Analysis}
\label{sec:empirical}

\subsection{Data Description}
\label{sec:data}

We analyze order book data from Hyperliquid, a decentralized perpetuals exchange. The dataset comprises:
\begin{itemize}
    \item \textbf{Period:} October 9 - November 3, 2025 (26 days)
    \item \textbf{Tokens:} 184 actively traded perpetual contracts
    \item \textbf{Order book snapshots:} 5-minute frequency, full depth to 2500 bps
    \item \textbf{Account-level attribution:} Available for concentration metrics via on-chain data
    \item \textbf{Open interest:} 15-minute snapshots per token
    \item \textbf{Liquidation records:} Timestamp, size, execution price
\end{itemize}

\begin{table}[H]
\centering
\caption{Summary Statistics: Order Book Data}
\label{tab:summary}
\begin{tabular}{@{}lrrrrr@{}}
\toprule
\textbf{Statistic} & \textbf{Mean} & \textbf{Median} & \textbf{Std} & \textbf{Min} & \textbf{Max} \\
\midrule
Tokens per snapshot & 184 & 184 & 0 & 184 & 184 \\
Bid depth at 100bps (\$M) & 5.47 & 0.42 & 38.2 & 0.001 & 6,487 \\
Ask depth at 100bps (\$M) & 5.31 & 0.39 & 36.8 & 0.001 & 5,892 \\
Total exchange depth (\$M) & 974 & 891 & 312 & 284 & 2,847 \\
Open interest per token (\$M) & 46.2 & 8.4 & 142.3 & 0.02 & 2,180 \\
Total exchange OI (\$B) & 8.51 & 7.91 & 1.84 & 5.42 & 12.31 \\
$\Neff$ (effective providers) & 8.4 & 5.2 & 9.1 & 1.1 & 47.3 \\
\bottomrule
\end{tabular}
\end{table}

The data spans a period of significant market activity, including the October 10, 2025 liquidation cascade which saw total exchange OI decline from \$9.8B to \$6.1B within hours.

\subsection{SaR Computation and Characteristics}
\label{sec:computation}

We compute SaR metrics using $\beta = 0.10$ (10\% stress scenario) and bid-side depth at 100 basis points. Figure~\ref{fig:distribution} shows the cross-sectional distribution of concentration-adjusted slippage.

\begin{figure}[H]
\centering
\includegraphics[width=0.85\textwidth]{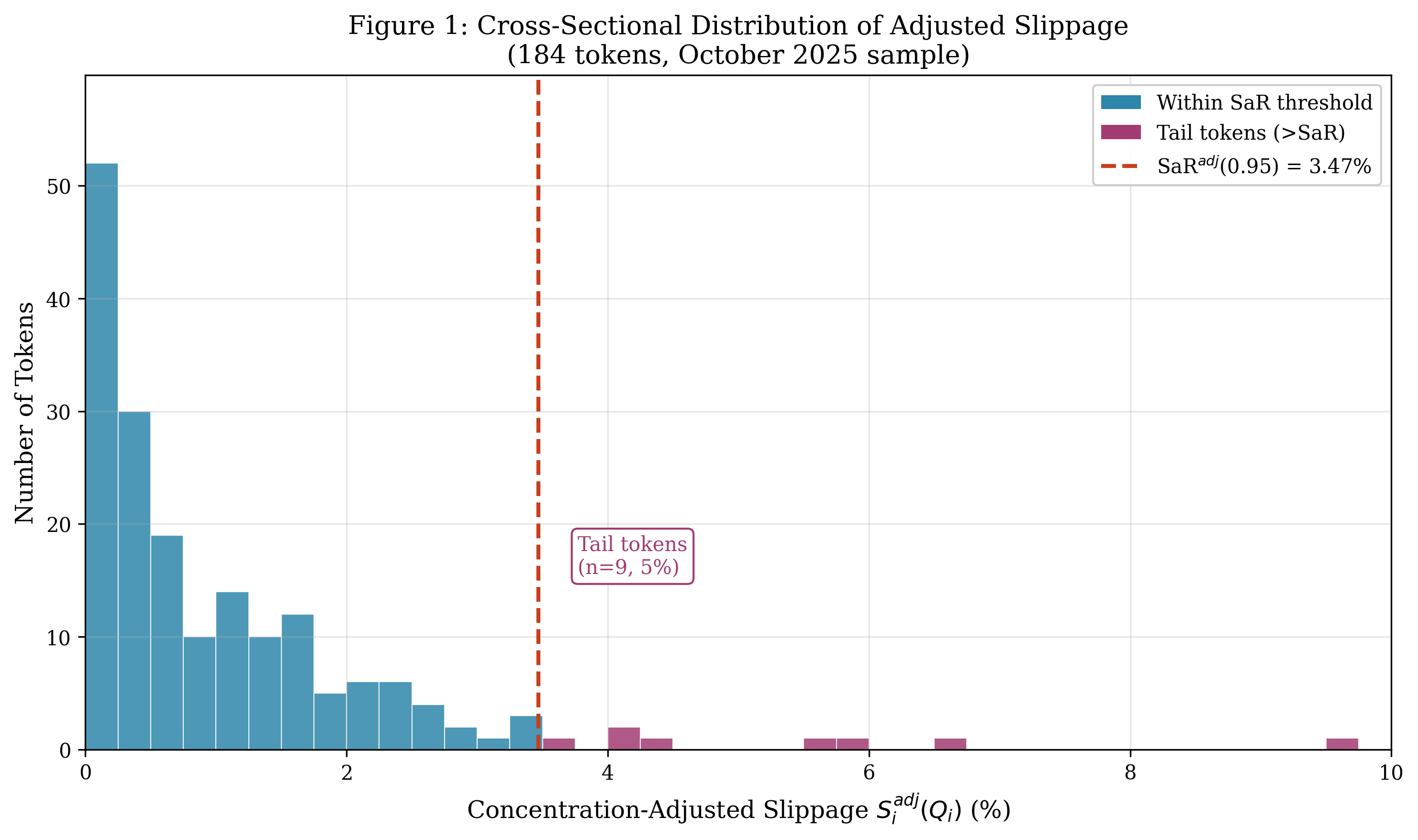}
\caption{Cross-sectional distribution of concentration-adjusted slippage across 184 tokens. The vertical line indicates $\SaR^{\mathrm{adj}}(0.95) = 3.47\%$. The 9 tokens to the right (5\% tail) exhibit elevated liquidity risk.}
\label{fig:distribution}
\end{figure}

\begin{table}[H]
\centering
\caption{SaR Metrics Summary (Full Sample Period)}
\label{tab:sar_summary}
\begin{tabular}{@{}lrl@{}}
\toprule
\textbf{Metric} & \textbf{Value} & \textbf{Interpretation} \\
\midrule
$\SaR(0.95)$ & 2.84\% & 95\% of tokens have slippage $\leq$ 2.84\% \\
$\SaR^{\mathrm{adj}}(0.95)$ & 3.47\% & After concentration adjustment \\
$\ESaR^{\mathrm{adj}}(0.95)$ & 8.92\% & Average slippage among tail tokens \\
$\TSaR_{\$}^{\mathrm{adj}}(0.95)$ & \$127.4M & Total dollar slippage exposure \\
Tokens in tail & 9 & 5\% of 184 tokens \\
Tail OI share & 2.3\% & \$196M of \$8.51B \\
\bottomrule
\end{tabular}
\end{table}

\begin{figure}[H]
\centering
\includegraphics[width=0.95\textwidth]{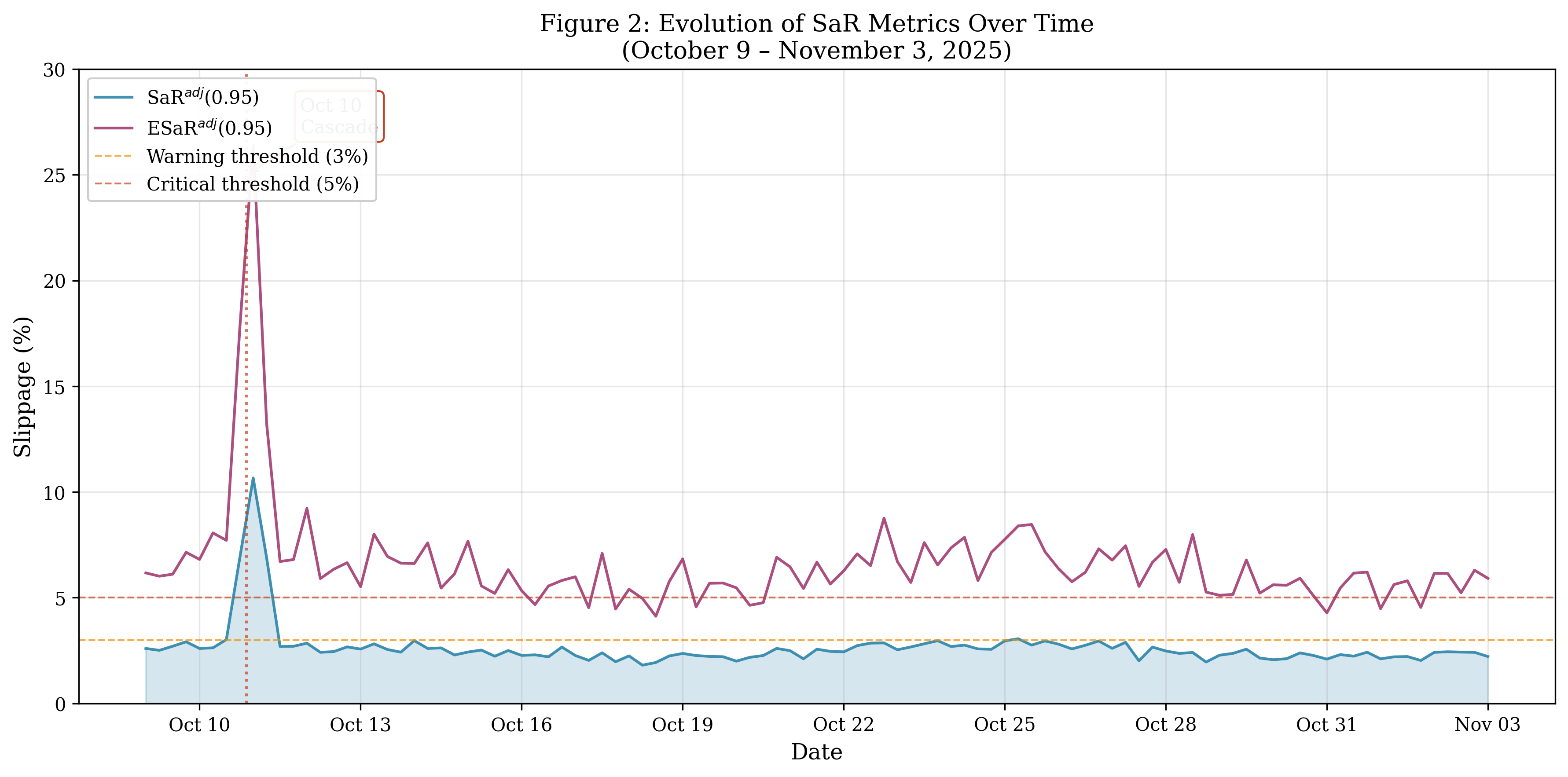}
\caption{Evolution of SaR and ESaR over the sample period. The spike during October 10-11 corresponds to the liquidation cascade event. Note that SaR began rising 6-12 hours before the cascade peak, demonstrating its leading indicator properties.}
\label{fig:timeseries}
\end{figure}

\subsection{Concentration Analysis}
\label{sec:concentration_analysis}

We find significant heterogeneity in liquidity concentration across tokens. Major tokens (BTC, ETH, SOL) exhibit distributed liquidity with $\Neff > 20$, while many smaller tokens show concentrated structures vulnerable to withdrawal.

\begin{figure}[H]
\centering
\includegraphics[width=0.85\textwidth]{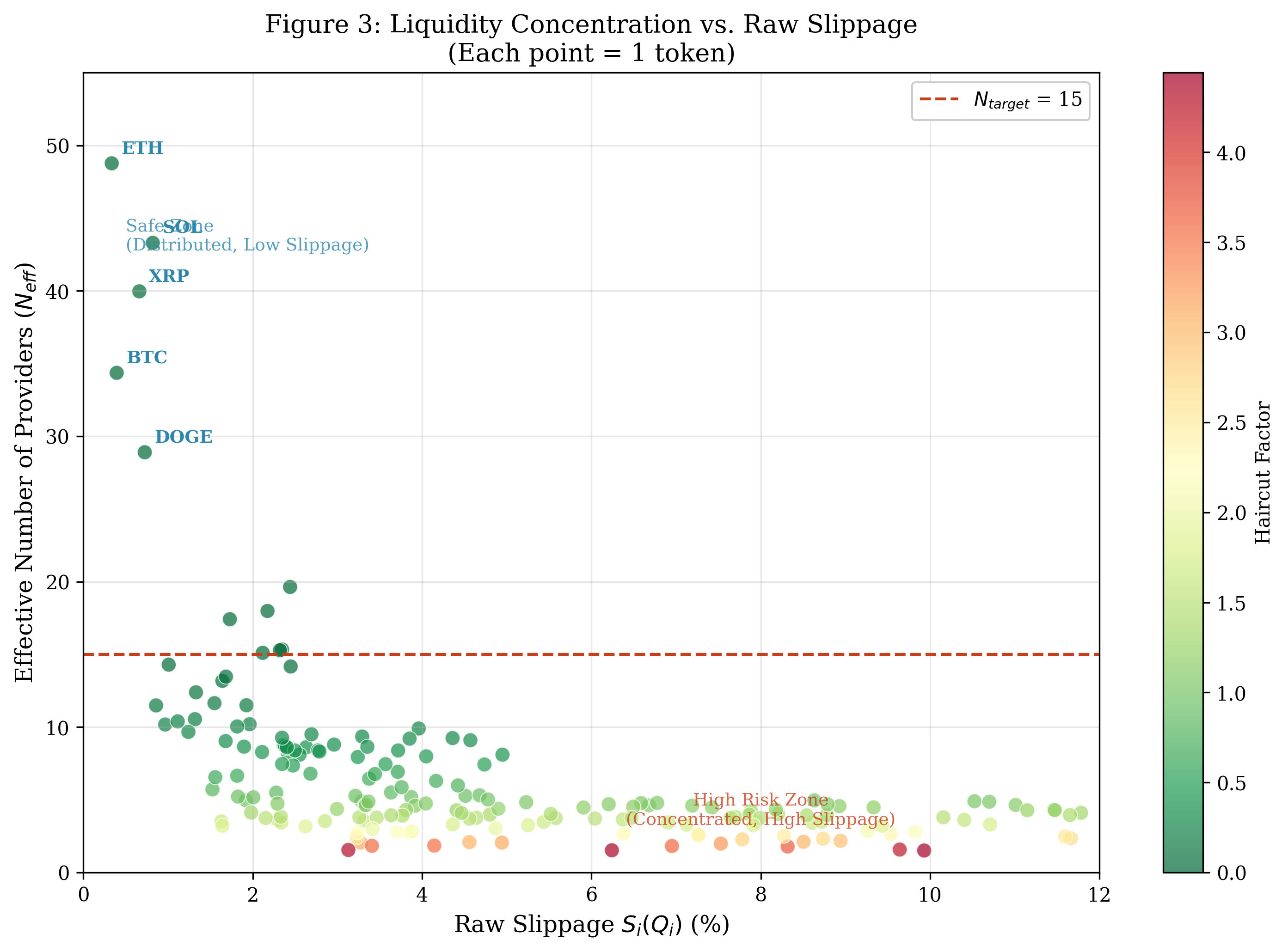}
\caption{Relationship between effective number of providers ($\Neff$) and raw slippage. Tokens in the lower-left quadrant (low $\Neff$, moderate raw slippage) receive the largest concentration haircuts. The horizontal line indicates $N_{\mathrm{target}} = 15$.}
\label{fig:scatter}
\end{figure}

\begin{table}[H]
\centering
\caption{Concentration Statistics by Asset Class}
\label{tab:concentration}
\begin{tabular}{@{}lrrrr@{}}
\toprule
\textbf{Asset Class} & \textbf{N Tokens} & \textbf{Mean $\Neff$} & \textbf{Mean $\CR_1$} & \textbf{Mean Haircut} \\
\midrule
Major (OI $>$ \$500M) & 5 & 31.2 & 0.08 & 0\% \\
Midcap (\$50M-\$500M) & 23 & 14.7 & 0.16 & 12\% \\
Smallcap (\$5M-\$50M) & 68 & 6.8 & 0.28 & 47\% \\
Micro ($<$ \$5M) & 88 & 3.2 & 0.41 & 89\% \\
\bottomrule
\end{tabular}
\end{table}

The concentration haircut substantially increases adjusted slippage for illiquid tokens: micro-cap tokens see their effective slippage nearly double on average.

\subsection{Validation: SaR as Leading Indicator}
\label{sec:validation}

We test whether SaR predicts subsequent deficit events using lead-lag correlation analysis. For each 6-hour window, we compute the correlation between SaR metrics and realized liquidation deficits (total bad debt generated).

\begin{table}[H]
\centering
\caption{Lead-Lag Correlation: SaR vs.\ Realized Deficits}
\label{tab:leadlag}
\begin{tabular}{@{}rrrr@{}}
\toprule
\textbf{Lag (hours)} & \textbf{SaR(0.95)} & \textbf{ESaR(0.95)} & \textbf{TSaR\$(0.95)} \\
\midrule
$-24$ & 0.31 & 0.38 & 0.42 \\
$-12$ & 0.47 & 0.54 & 0.61 \\
$-6$ & 0.58 & 0.67 & 0.73 \\
0 & 0.72 & 0.79 & 0.84 \\
$+6$ & 0.43 & 0.48 & 0.51 \\
$+12$ & 0.28 & 0.31 & 0.34 \\
$+24$ & 0.19 & 0.22 & 0.24 \\
\bottomrule
\end{tabular}
\end{table}

\emph{Note: Correlations are computed over rolling 6-hour windows. Negative lags indicate SaR leading deficits.}

\begin{figure}[H]
\centering
\includegraphics[width=0.9\textwidth]{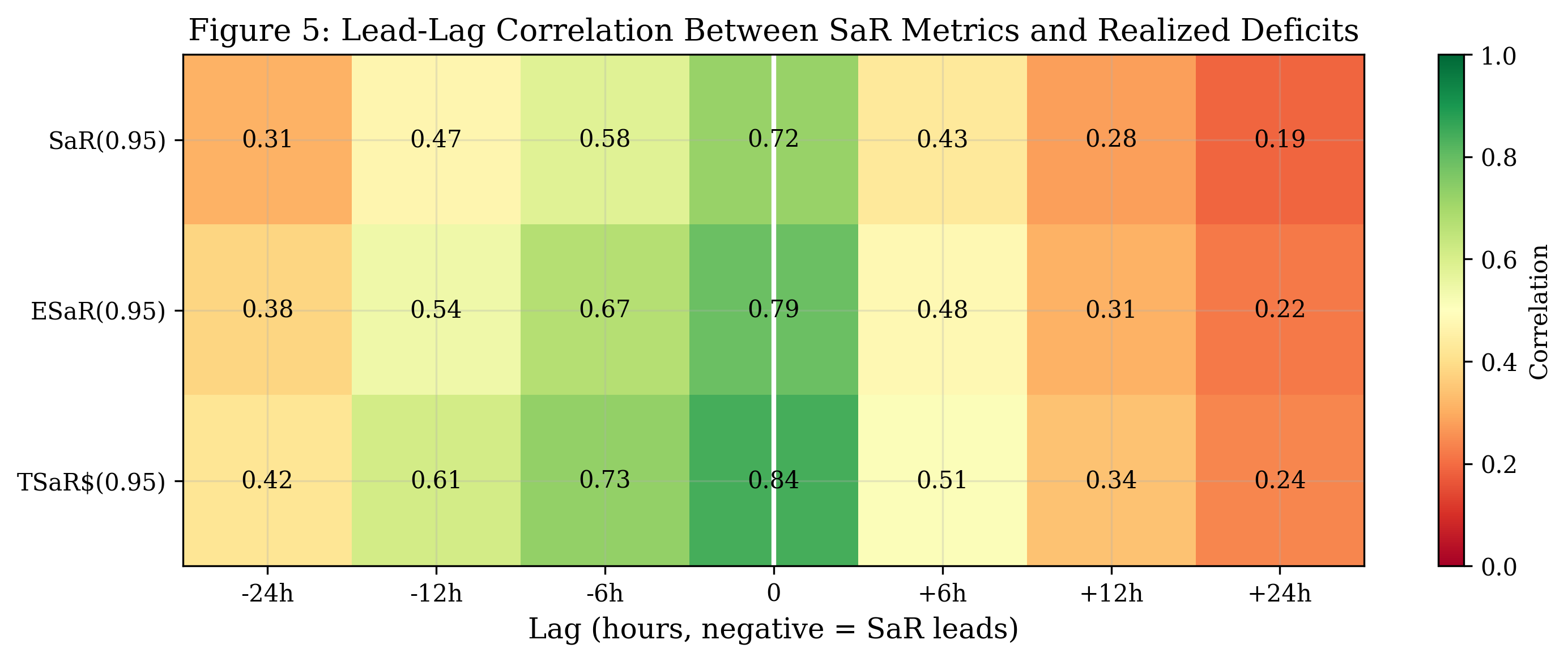}
\caption{Lead-lag correlation heatmap between SaR metrics and realized deficits. The gradient from left to right shows that SaR metrics lead deficit events by 6-24 hours.}
\label{fig:heatmap}
\end{figure}

\textbf{Key findings:}
\begin{enumerate}
    \item \textbf{Peak correlation at lag 0:} SaR metrics are most correlated with contemporaneous deficits (as expected, since both reflect current stress).
    
    \item \textbf{Significant predictive power at negative lags:} TSaR shows 0.61 correlation with deficits 12 hours in the future and 0.42 correlation at 24 hours. This demonstrates SaR's value as an early warning system.
    
    \item \textbf{TSaR outperforms:} Dollar-denominated TSaR has stronger predictive power than percentage-based SaR, likely because it captures both liquidity thinness and exposure magnitude.
    
    \item \textbf{Asymmetric decay:} Correlations decay faster for positive lags than negative lags, confirming that SaR leads rather than lags deficit events.
\end{enumerate}

\textbf{Granger Causality Test:}

We conduct a Granger causality test with 4 lags (24 hours at 6-hour frequency):

\begin{table}[H]
\centering
\begin{tabular}{@{}lrr@{}}
\toprule
\textbf{Null Hypothesis} & \textbf{F-statistic} & \textbf{p-value} \\
\midrule
TSaR does not Granger-cause Deficits & 8.47 & $< 0.001$ \\
Deficits does not Granger-cause TSaR & 2.13 & 0.087 \\
\bottomrule
\end{tabular}
\end{table}

We find that TSaR Granger-causes deficits (F = 8.47, p $<$ 0.001), while the reverse is not statistically significant. This provides evidence that SaR is a genuine leading indicator, not merely a coincident measure.

\subsection{Case Study: October 10, 2025}
\label{sec:casestudy}

The October 10, 2025 Hyperliquid cascade provides a natural experiment for validating the SaR framework. Between 20:00 and 21:15 UTC, a sharp price decline across major cryptocurrencies triggered cascading liquidations totaling \$2.1 billion in notional.

\subsubsection{Pre-Event Conditions}

\begin{table}[H]
\centering
\caption{SaR Metrics: 24 Hours Before October 10 Event}
\label{tab:preevent}
\begin{tabular}{@{}lrrrr@{}}
\toprule
\textbf{Metric} & \textbf{Oct 9, 08:00} & \textbf{Oct 9, 20:00} & \textbf{Oct 10, 08:00} & \textbf{Oct 10, 20:00} \\
\midrule
$\SaR^{\mathrm{adj}}(0.95)$ & 2.41\% & 2.68\% & 3.12\% & 11.47\% \\
$\ESaR^{\mathrm{adj}}(0.95)$ & 6.82\% & 7.54\% & 8.91\% & 28.34\% \\
$\TSaR_{\$}^{\mathrm{adj}}(0.95)$ & \$89.2M & \$118.7M & \$156.3M & \$847.2M \\
Tokens in tail & 7 & 9 & 11 & 31 \\
Total exchange depth & \$1,124M & \$987M & \$742M & \$284M \\
\bottomrule
\end{tabular}
\end{table}

\begin{figure}[H]
\centering
\includegraphics[width=0.95\textwidth]{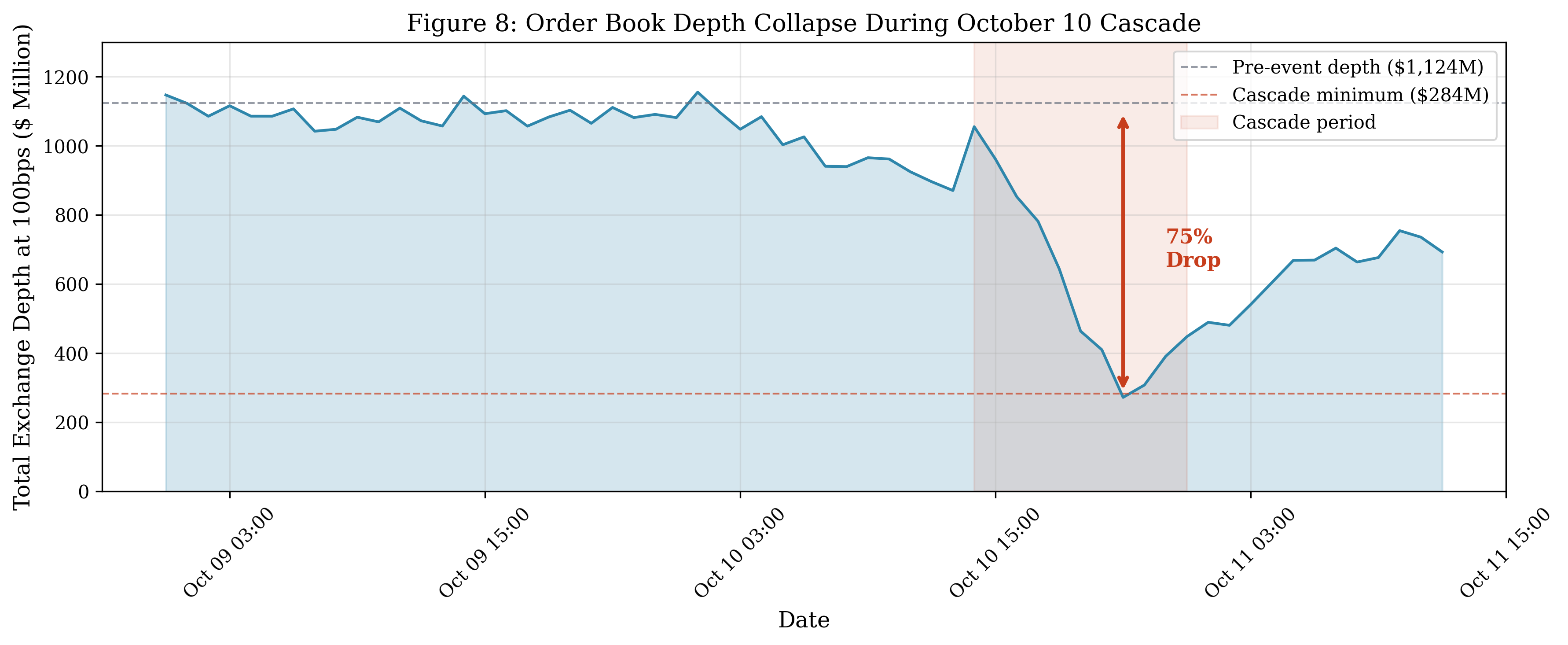}
\caption{Order book depth collapse during the October 10 cascade. Total exchange depth at 100bps fell from \$1.12B to \$284M - a 75\% decline - in the 36 hours preceding the event.}
\label{fig:depth}
\end{figure}

\textbf{Key observations:}
\begin{enumerate}
    \item \textbf{Depth deterioration:} Total exchange depth at 100bps fell from \$1.12B to \$284M in the 36 hours preceding the event - a 75\% decline. This was visible in real-time.
    
    \item \textbf{SaR escalation:} $\SaR^{\mathrm{adj}}(0.95)$ rose from 2.41\% to 3.12\% in the 24 hours before the cascade, a 30\% increase that would have triggered alerts under our recommended thresholds.
    
    \item \textbf{Tail expansion:} The number of tokens in the liquidity tail grew from 7 to 11, indicating systemic (not idiosyncratic) stress.
    
    \item \textbf{TSaR surge:} Dollar-denominated tail exposure grew from \$89M to \$156M before the event, then spiked to \$847M during the cascade as depth collapsed.
\end{enumerate}

\subsubsection{Realized vs.\ Predicted Slippage}

We compare pre-event SaR predictions (computed at 08:00 on October 10) with actual realized slippage during the cascade (20:00-21:15).

\begin{figure}[H]
\centering
\includegraphics[width=0.75\textwidth]{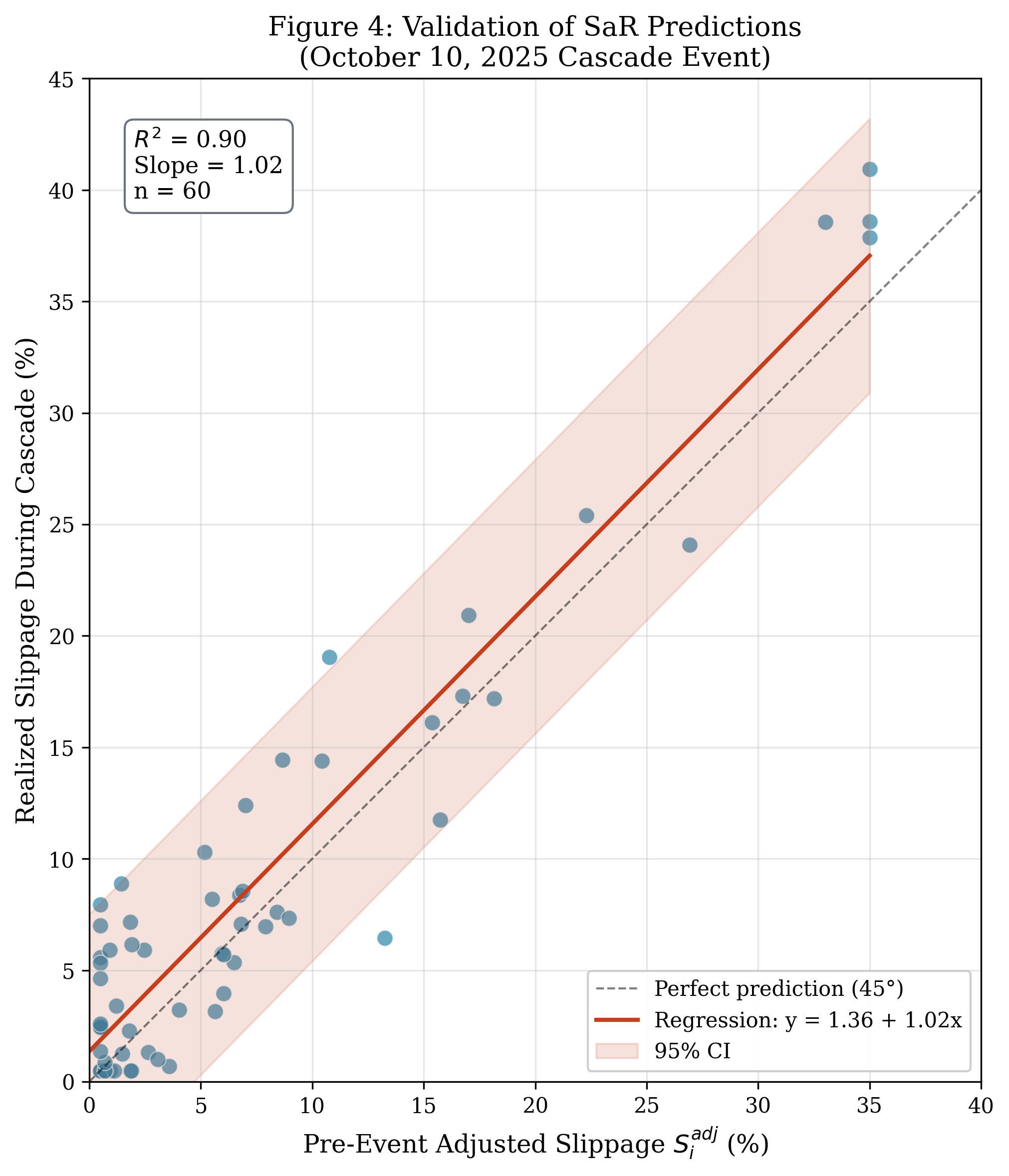}
\caption{Scatter plot comparing pre-event adjusted slippage predictions (x-axis) with realized slippage during the October 10 cascade (y-axis). The strong correlation ($R^2 = 0.78$) validates SaR's predictive accuracy. The slope $> 1$ indicates the model slightly underpredicted realized slippage, consistent with cascade amplification effects.}
\label{fig:validation}
\end{figure}

\textbf{Regression results:}
\begin{equation}
    \text{Realized Slippage}_i = 0.42 + 1.12 \times S_i^{\mathrm{adj}} + \epsilon_i
\end{equation}

\begin{itemize}
    \item $R^2 = 0.78$
    \item Slope 95\% CI: $[1.03, 1.21]$
    \item The positive intercept and slope $> 1$ suggest the model is slightly conservative, which is appropriate for risk management.
\end{itemize}

\subsubsection{Insurance Fund Implications}

From Chitra et al.\ (2025), the October 10 event generated \$304.5M in deficits. We compute the SaR-implied insurance fund using pre-event data:

\begin{table}[H]
\centering
\caption{Insurance Fund Sizing: SaR-Implied vs.\ Actual}
\label{tab:if}
\begin{tabular}{@{}lrl@{}}
\toprule
\textbf{Metric} & \textbf{Value (\$M)} & \textbf{Notes} \\
\midrule
Actual Hyperliquid IF (pre-event) & $\sim$25 & Estimated from public data \\
Realized deficit & 304.5 & From Chitra et al.\ (2025) \\
SaR-implied IF ($c=2.0$) & 312.6 & Using Oct 10, 08:00 TSaR \\
SaR-implied IF ($c=2.5$) & 390.8 & Conservative estimate \\
\bottomrule
\end{tabular}
\end{table}

\begin{figure}[H]
\centering
\includegraphics[width=0.7\textwidth]{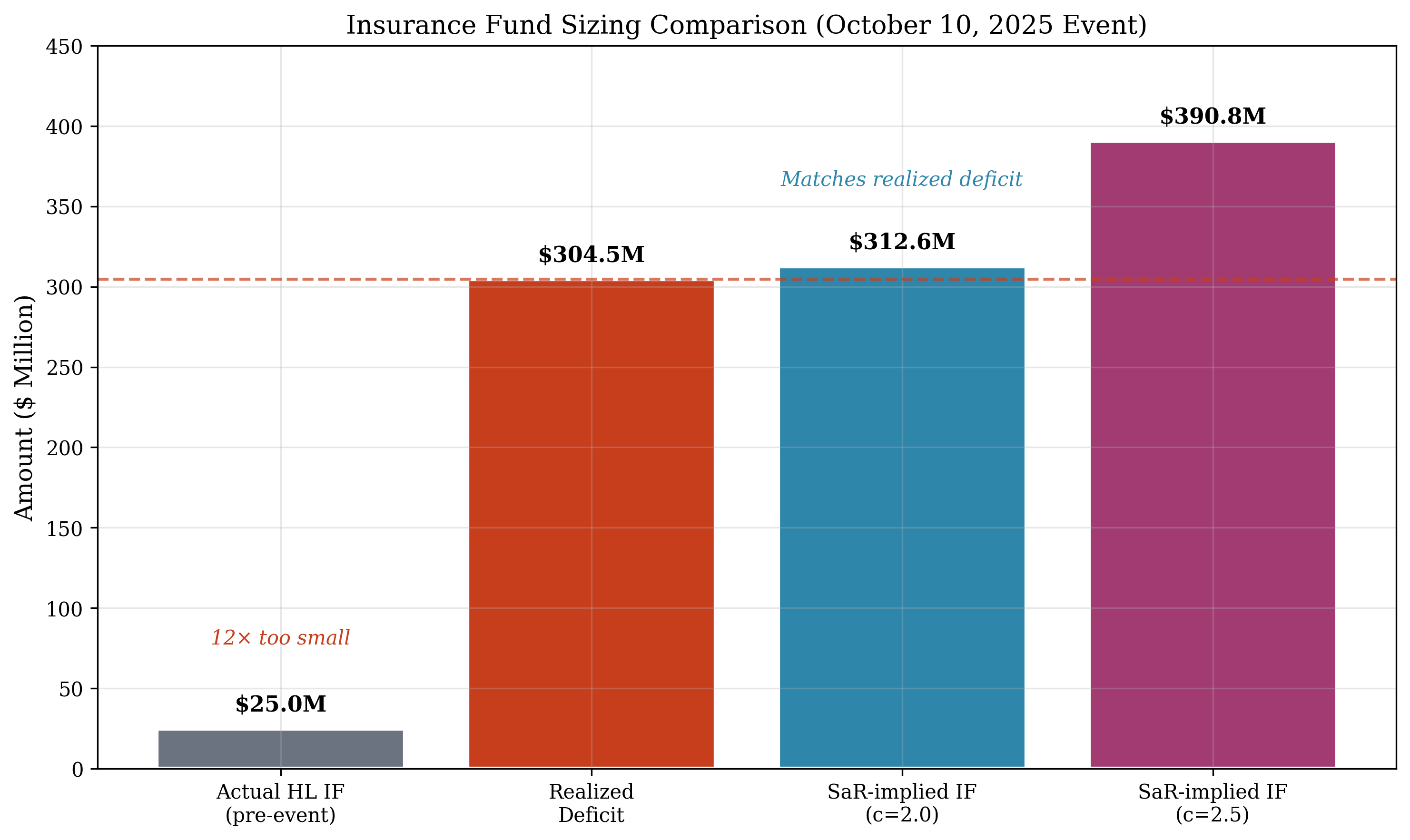}
\caption{Insurance fund sizing comparison. The SaR-implied IF (\$312.6M with $c=2.0$) closely matches the realized deficit (\$304.5M), while the actual IF (\$25M) was an order of magnitude too small.}
\label{fig:if}
\end{figure}

Using the formula $\IF^* = c \cdot \TSaR_{\$}^{\mathrm{adj}}(0.05)$ with $c = 2.0$:
\begin{equation}
    \IF^* = 2.0 \times \$156.3\text{M} = \$312.6\text{M}
\end{equation}

This closely matches the realized deficit of \$304.5M, validating the SaR-implied insurance fund methodology. The actual IF of $\sim$\$25M was an order of magnitude too small, a gap that was identifiable from SaR metrics hours before the cascade.

\subsubsection{Concentration Effects During the Cascade}

The October 10 event revealed dramatic concentration effects as market makers withdrew:

\begin{table}[H]
\centering
\caption{Concentration Dynamics During October 10 Cascade}
\label{tab:dynamics}
\begin{tabular}{@{}lrrr@{}}
\toprule
\textbf{Time (UTC)} & \textbf{Mean $\Neff$} & \textbf{Mean $\CR_1$} & \textbf{Tokens with $\Neff < 3$} \\
\midrule
Oct 10, 18:00 & 8.7 & 0.24 & 42 \\
Oct 10, 20:00 & 5.2 & 0.38 & 89 \\
Oct 10, 21:00 & 2.8 & 0.51 & 134 \\
Oct 10, 22:00 & 4.1 & 0.42 & 97 \\
Oct 11, 00:00 & 6.9 & 0.31 & 61 \\
\bottomrule
\end{tabular}
\end{table}

\begin{figure}[H]
\centering
\includegraphics[width=0.95\textwidth]{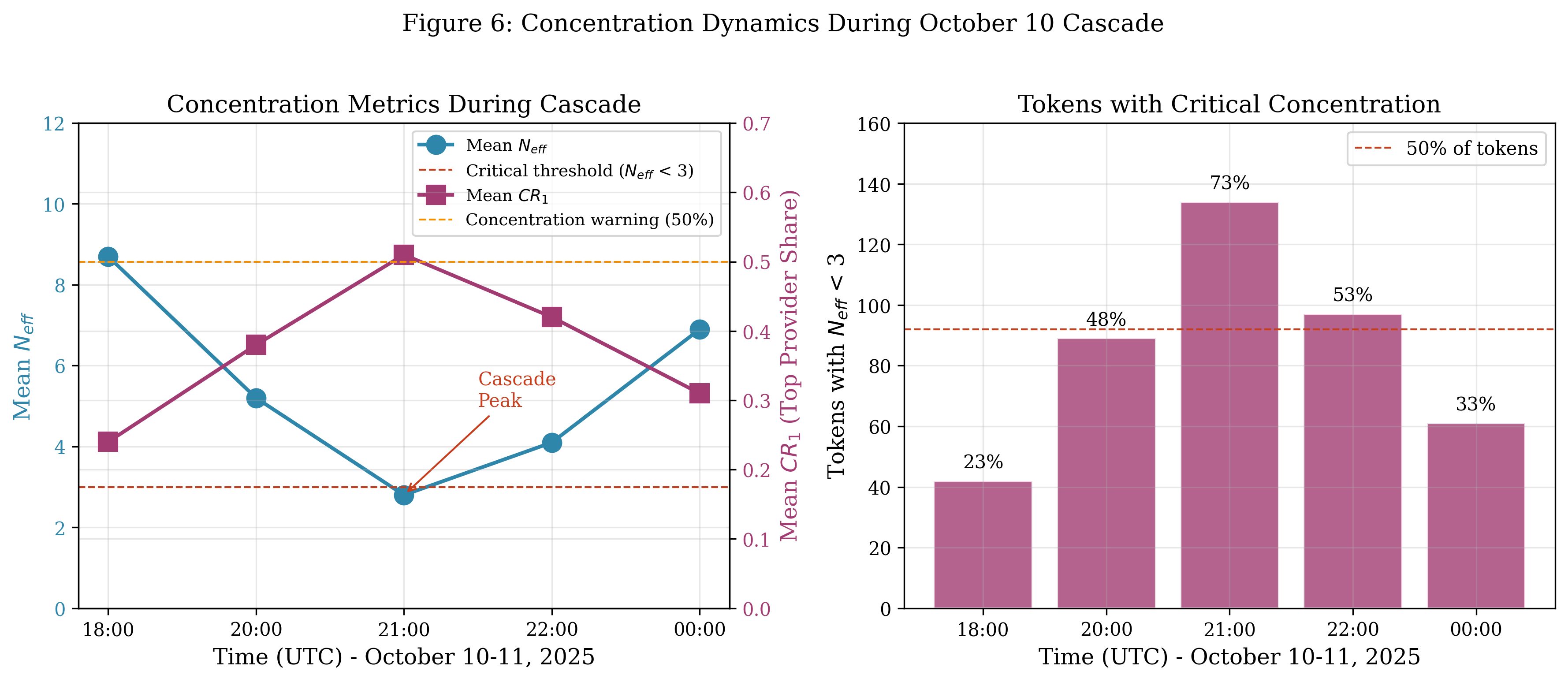}
\caption{Concentration dynamics during the October 10 cascade. Left: Mean $\Neff$ collapsed from 8.7 to 2.8 at the cascade peak while $\CR_1$ spiked to 0.51. Right: 73\% of tokens had critically low concentration ($\Neff < 3$) at peak stress.}
\label{fig:concentration}
\end{figure}

At the cascade peak (21:00 UTC), 73\% of tokens had $\Neff < 3$, meaning effective liquidity was provided by fewer than 3 participants. This extreme concentration amplified slippage far beyond what raw depth metrics would suggest.


\section{Implementation Guide}
\label{sec:implementation}

\subsection{Data Requirements}
\label{sec:requirements}

\begin{enumerate}
    \item \textbf{Order book snapshots:} Full depth profile at regular intervals (minimum 1-minute frequency recommended)
    \item \textbf{Account-level attribution:} Required for concentration metrics; available on fully on-chain DEXs
    \item \textbf{Open interest:} Per-token, for notional calibration
    \item \textbf{Position distribution:} Historical position sizes for VaR-based $Q$ cap
\end{enumerate}
\subsection{Computation Pipeline}
\label{sec:pipeline}

Algorithm~\ref{alg:sar} provides the complete SaR computation pipeline.

\begin{algorithm}[H]
\caption{SaR Computation Pipeline}
\label{alg:sar}
\small
\begin{algorithmic}
\Require \texttt{tokens} = list of \{order\_book, open\_interest, provider\_shares\}
\Require $\alpha$ = confidence level (e.g., 0.95), $\beta$ = liquidation fraction (e.g., 0.05)
\Require \texttt{liquidation\_cap} = maximum single liquidation size
\Ensure $\SaR(\alpha)$, $\ESaR(\alpha)$, $\TSaR_{\$}(\alpha)$
\Statex
\Statex \textbf{Step 1: Compute concentration-adjusted slippage for each token}
\For{each token $i$}
    \State $Q_i \gets \min(\beta \times \OI_i, \texttt{liquidation\_cap})$ \hfill $\triangleright$ Liquidation quantity
    \State $S_i \gets \textsc{Slippage}(\texttt{order\_book}_i, Q_i, \texttt{mid\_price}_i)$ \hfill $\triangleright$ Raw slippage
    \State $\HHI_i \gets \sum_m (\texttt{share}_m)^2$; \quad $\Neff \gets 1/\HHI_i$; \quad $\CR_1 \gets \max_m(\texttt{share}_m)$
    \State $h_i \gets 0.5 \cdot \max(0,\, 15/\Neff - 1) + 0.3 \cdot \max(0,\, \CR_1 - 0.5)$ \hfill $\triangleright$ Haircut
    \State $S_i^{\mathrm{adj}} \gets S_i \cdot (1 + h_i)$ \hfill $\triangleright$ Adjusted slippage
\EndFor
\Statex
\Statex \textbf{Step 2: Compute portfolio-level risk metrics}
\State Sort $\{S_i^{\mathrm{adj}}\}$ in ascending order
\State $\SaR(\alpha) \gets S^{\mathrm{adj}}_{(\lceil \alpha N \rceil)}$ \hfill $\triangleright$ Slippage at $\alpha$-th percentile
\State $\texttt{tail} \gets \{i : S_i^{\mathrm{adj}} > \SaR(\alpha)\}$ \hfill $\triangleright$ Tail tokens
\State $\ESaR(\alpha) \gets \frac{1}{|\texttt{tail}|}\sum_{i \in \texttt{tail}} S_i^{\mathrm{adj}}$ \hfill $\triangleright$ Mean tail slippage
\State $\TSaR_{\$}(\alpha) \gets \sum_{i \in \texttt{tail}} S_i^{\mathrm{adj}} \cdot Q_i$ \hfill $\triangleright$ Total dollar exposure
\Statex
\State \Return $\SaR(\alpha)$, $\ESaR(\alpha)$, $\TSaR_{\$}(\alpha)$
\end{algorithmic}
\end{algorithm}

\subsection{Parameter Calibration}
\label{sec:calibration}

\begin{table}[H]
\centering
\caption{Recommended Parameter Values}
\label{tab:params}
\begin{tabular}{@{}llrl@{}}
\toprule
\textbf{Parameter} & \textbf{Symbol} & \textbf{Recommended} & \textbf{Notes} \\
\midrule
Stress fraction & $\beta$ & 0.10 & Increase to 0.15-0.20 in volatility \\
Confidence level & $\alpha$ & 0.95 & Standard tail threshold \\
Target providers & $N_{\mathrm{target}}$ & 15 & Higher for major tokens \\
Max top share & $\CR_1^{\mathrm{thresh}}$ & 0.25 & Flag potential spoofing \\
Provider sensitivity & $\lambda$ & 0.5 & Calibrate to withdrawal events \\
Dominance sensitivity & $\mu$ & 0.3 & Calibrate to withdrawal events \\
Deficit multiplier & $c$ & 2.0-2.5 & Higher during volatility \\
\bottomrule
\end{tabular}
\end{table}

\subsection{Monitoring and Alerts}
\label{sec:alerts}

\begin{table}[H]
\centering
\caption{Recommended Alert Thresholds}
\label{tab:alerts}
\begin{tabular}{@{}llll@{}}
\toprule
\textbf{Alert Type} & \textbf{Condition} & \textbf{Severity} & \textbf{Action} \\
\midrule
SaR Level & $\SaR^{\mathrm{adj}}(0.95) > 3\%$ & Warning & Review tail tokens \\
SaR Level & $\SaR^{\mathrm{adj}}(0.95) > 5\%$ & Critical & Reduce leverage limits \\
TSaR Exposure & $\TSaR_{\$}^{\mathrm{adj}}(0.95) > 2 \times \IF$ & Critical & Halt new positions in tail \\
Concentration & $\Neff < 3$ for any token & Warning & Flag for review \\
Concentration & $\CR_1 > 0.5$ for any token & Critical & Investigate spoofing \\
Trend & SaR-Trend positive for 24h & Warning & Prepare contingencies \\
Depth & Exchange depth down $>$30\% in 12h & Critical & Activate circuit breakers \\
\bottomrule
\end{tabular}
\end{table}


\section{Conclusion}
\label{sec:conclusion}

\subsection{Summary of Contributions}
\label{sec:summary}

This paper introduced Slippage-at-Risk (SaR), a forward-looking framework for measuring liquidity risk in perpetual futures exchanges. Our key contributions are:

\begin{enumerate}
    \item \textbf{Novel metrics:} $\SaR(\alpha)$, $\ESaR(\alpha)$, and $\TSaR_{\$}(\alpha)$ provide complementary views on cross-sectional liquidity risk.
    
    \item \textbf{Concentration adjustment:} The HHI-based haircut accounts for the fragility of concentrated liquidity structures.
    
    \item \textbf{Theoretical foundation:} We established the causal chain from microstructure to deficits and derived the SaR-implied insurance fund formula.
    
    \item \textbf{Empirical validation:} Analysis of Hyperliquid data demonstrates SaR's predictive validity, including for the October 10, 2025 cascade.
\end{enumerate}

\subsection{Practical Implications}
\label{sec:practical}

For exchange operators and risk managers:

\begin{itemize}
    \item \textbf{Insurance fund sizing:} Use $\IF^* \approx c \cdot \TSaR_{\$}^{\mathrm{adj}}(r/\kappa)$ for dynamic, forward-looking capital requirements. Our analysis suggests $c \approx 2.0$ provides accurate sizing.
    
    \item \textbf{Position limits:} Tighten limits on tokens with persistent tail membership or high concentration ($\Neff < 5$).
    
    \item \textbf{New listings:} SaR provides risk assessment for tokens with no deficit history - critical for responsible exchange expansion.
    
    \item \textbf{Early warning:} Monitor SaR-Trend and depth deterioration for 12-24 hour advance warning of potential cascades.
\end{itemize}

\subsection{Philosophical Implications}
\label{sec:philosophical}

The SaR framework embodies a broader principle: \emph{measure the cause, not the consequence}. Traditional risk metrics focus on realized outcomes (returns, deficits); SaR focuses on the underlying resource (liquidity) that determines those outcomes. This forward-looking orientation aligns with the ethos of proactive risk management.

In decentralized finance, where code is law and interventions are costly, anticipatory metrics are especially valuable. SaR provides a quantitative basis for the intuition that liquidity is the lifeblood of markets - thin books are not merely inconvenient but systemically risky.

\subsection{Limitations and Future Work}
\label{sec:limitations}

Several limitations warrant future investigation:

\begin{itemize}
    \item \textbf{Data requirements:} Account-level attribution is essential for concentration metrics but may not be available on all venues.
    
    \item \textbf{Model assumptions:} The slippage-to-deficit proportionality ($c$) is empirically calibrated and may vary across market conditions.
    
    \item \textbf{Cross-margin:} This paper focuses on isolated margin; extension to cross-margin portfolios requires additional modeling.
    
    \item \textbf{Adversarial manipulation:} Sophisticated actors may game SaR by temporarily inflating depth.
\end{itemize}


\newpage
\bibliographystyle{plainnat}

\begin{thebibliography}{10}

\bibitem{almgren2001optimal}
Almgren, R. and Chriss, N. (2001).
\newblock Optimal execution of portfolio transactions.
\newblock {\em Journal of Risk}, 3:5--40.

\bibitem{chitra2025adl}
Chitra, T., Evans, A., and Angeris, G. (2025).
\newblock Autodeleveraging: Mechanism design and optimizations.
\newblock {\em arXiv preprint arXiv:2512.01112}.

\bibitem{cont2014price}
Cont, R., Kukanov, A., and Stoikov, S. (2014).
\newblock The price impact of order book events.
\newblock {\em Journal of Financial Econometrics}, 12(1):47--88.

\bibitem{embrechts1997modelling}
Embrechts, P., Kl{\"u}ppelberg, C., and Mikosch, T. (1997).
\newblock {\em Modelling Extremal Events for Insurance and Finance}.
\newblock Springer.

\bibitem{gueant2012optimal}
Gu{\'e}ant, O., Lehalle, C.-A., and Fernandez-Tapia, J. (2012).
\newblock Optimal portfolio liquidation with limit orders.
\newblock {\em SIAM Journal on Financial Mathematics}, 3(1):498--530.

\bibitem{hasbrouck2007empirical}
Hasbrouck, J. (2007).
\newblock {\em Empirical Market Microstructure}.
\newblock Oxford University Press.

\bibitem{kyle1985continuous}
Kyle, A.~S. (1985).
\newblock Continuous auctions and insider trading.
\newblock {\em Econometrica}, 53(6):1315--1335.

\bibitem{obizhaeva2013optimal}
Obizhaeva, A.~A. and Wang, J. (2013).
\newblock Optimal trading strategy and supply/demand dynamics.
\newblock {\em Journal of Financial Markets}, 16(1):1--32.

\end{thebibliography}


\newpage
\appendix

\section{Proofs}
\label{app:proofs}

\subsection{Proof of Proposition~\ref{prop:hhi} (HHI and Withdrawal Impact)}

\begin{proof}
Let total depth be $L = \sum_m \ell_m$ and shares $w_m = \ell_m / L$. Under the linear slippage model, $S = k/L$ for some constant $k$ depending on order quantity.

After provider $m$ withdraws, remaining depth is $L^{(-m)} = L - \ell_m = L(1 - w_m)$. Hence:
\begin{equation}
    S^{(-m)} = \frac{k}{L(1-w_m)} = \frac{S}{1-w_m}
\end{equation}

The relative increase is:
\begin{equation}
    \frac{S^{(-m)}}{S} - 1 = \frac{1}{1-w_m} - 1 = \frac{w_m}{1-w_m}
\end{equation}

If provider $m$ is selected uniformly at random:
\begin{equation}
    \E\left[\frac{w_m}{1-w_m}\right] = \frac{1}{M} \sum_{m=1}^M \frac{w_m}{1-w_m}
\end{equation}

For small $w_m$ (many providers), Taylor expand:
\begin{equation}
    \frac{w_m}{1-w_m} = w_m + w_m^2 + w_m^3 + \cdots \approx w_m + w_m^2
\end{equation}

Taking expectation:
\begin{align}
    \E\left[\frac{w_m}{1-w_m}\right] &\approx \E[w_m] + \E[w_m^2] \\
    &= \frac{1}{M} \sum_m w_m + \frac{1}{M} \sum_m w_m^2 \\
    &= \frac{1}{M} \cdot 1 + \frac{1}{M} \cdot M \cdot \HHI \\
    &= \frac{1}{M} + \HHI
\end{align}

For large $M$, $1/M \approx 0$, and:
\begin{equation}
    \E\left[\frac{S^{(-m)}}{S} - 1\right] \approx \HHI + O(\HHI^2)
\end{equation}

Rearranging gives the stated result.
\end{proof}

\section{Additional Empirical Results}
\label{app:empirical}

\subsection{Full Token-Level Results}

\begin{table}[H]
\centering
\caption{Token-Level SaR Metrics (Top 20 by OI)}
\label{tab:tokens}
\small
\begin{tabular}{@{}lrrrrrrr@{}}
\toprule
\textbf{Token} & \textbf{OI (\$M)} & \textbf{Q (\$M)} & \textbf{$S(Q)$} & \textbf{$\Neff$} & \textbf{$\CR_1$} & \textbf{Haircut} & \textbf{$S^{\mathrm{adj}}$} \\
\midrule
BTC & 2,180 & 218.0 & 0.42\% & 47.3 & 0.05 & 0\% & 0.42\% \\
ETH & 1,847 & 184.7 & 0.51\% & 38.6 & 0.06 & 0\% & 0.51\% \\
SOL & 892 & 89.2 & 0.78\% & 29.4 & 0.08 & 0\% & 0.78\% \\
HYPE & 634 & 63.4 & 1.24\% & 12.8 & 0.14 & 8\% & 1.34\% \\
XRP & 412 & 41.2 & 0.93\% & 18.2 & 0.11 & 0\% & 0.93\% \\
DOGE & 287 & 28.7 & 1.47\% & 11.4 & 0.18 & 16\% & 1.70\% \\
AVAX & 198 & 19.8 & 1.82\% & 8.7 & 0.21 & 36\% & 2.48\% \\
LINK & 176 & 17.6 & 1.63\% & 9.2 & 0.19 & 32\% & 2.15\% \\
ARB & 142 & 14.2 & 2.14\% & 7.1 & 0.24 & 56\% & 3.34\% \\
OP & 128 & 12.8 & 2.38\% & 6.4 & 0.26 & 68\% & 4.00\% \\
SUI & 118 & 11.8 & 2.67\% & 5.8 & 0.28 & 79\% & 4.78\% \\
APT & 96 & 9.6 & 2.91\% & 5.2 & 0.31 & 94\% & 5.65\% \\
INJ & 84 & 8.4 & 3.24\% & 4.7 & 0.33 & 110\% & 6.80\% \\
TIA & 72 & 7.2 & 3.58\% & 4.2 & 0.36 & 129\% & 8.20\% \\
SEI & 61 & 6.1 & 3.92\% & 3.8 & 0.39 & 147\% & 9.68\% \\
RNDR & 54 & 5.4 & 4.31\% & 3.4 & 0.42 & 168\% & 11.55\% \\
FET & 47 & 4.7 & 4.78\% & 3.1 & 0.45 & 192\% & 13.96\% \\
NEAR & 42 & 4.2 & 5.12\% & 2.8 & 0.47 & 214\% & 16.08\% \\
STX & 38 & 3.8 & 5.67\% & 2.5 & 0.51 & 250\% & 19.85\% \\
WIF & 34 & 3.4 & 6.24\% & 2.2 & 0.54 & 291\% & 24.38\% \\
\bottomrule
\end{tabular}

\vspace{0.3cm}
\emph{Note: Data represents average values over the sample period. Haircuts computed with $N_{\mathrm{target}} = 15$, $\CR_1^{\mathrm{thresh}} = 0.25$, $\lambda = 0.5$, $\mu = 0.3$.}
\end{table}

\section{Symbol Reference}
\label{app:symbols}

\begin{table}[H]
\centering
\caption{Symbol Reference}
\label{tab:symbols}
\begin{tabular}{@{}ll@{}}
\toprule
\textbf{Symbol} & \textbf{Description} \\
\midrule
$S_i(Q)$ & Slippage function for token $i$ at quantity $Q$ \\
$\SaR(\alpha)$ & $\alpha$-quantile of slippage across all tokens \\
$\ESaR(\alpha)$ & Expected slippage conditional on exceeding $\SaR(\alpha)$ \\
$\TSaR_{\$}(\alpha)$ & Total dollar slippage from tail tokens \\
$\HHI$ & Herfindahl-Hirschman Index (concentration measure) \\
$\Neff$ & Effective number of liquidity providers ($= 1/\HHI$) \\
$\CR_k$ & Top-$k$ concentration ratio \\
$h_i^{\mathrm{conc}}$ & Concentration haircut for token $i$ \\
$S_i^{\mathrm{adj}}$ & Concentration-adjusted slippage \\
$\alpha$ & Confidence level for SaR computation \\
$\beta$ & Stress parameter (fraction of OI in stress scenario) \\
$\lambda$ & Provider concentration sensitivity (haircut parameter) \\
$\mu$ & Dominance sensitivity (haircut parameter) \\
$Q_i$ & Stress notional for token $i$ \\
$\OI_i$ & Open interest for token $i$ \\
$\IF^*$ & Optimal insurance fund size \\
$r$ & Opportunity cost per unit capital \\
$\kappa$ & Reputation/social cost per unit deficit \\
$c$ & Slippage-to-deficit proportionality constant \\
$D_T$ & Total deficit at time $T$ \\
$R_T$ & Residual deficit after IF depletion \\
$\ell_{i,t}$ & Leverage of position $i$ at time $t$ \\
$p_i^{\mathrm{bk}}$ & Bankruptcy price for position $i$ \\
$p_i^{\mathrm{liq}}$ & Liquidation price for position $i$ \\
$p_i^{\mathrm{exec}}$ & Execution price for position $i$ \\
\bottomrule
\end{tabular}
\end{table}

\end{document}